%% file: main.tex
\documentclass[11pt,a4paper]{article}
\usepackage[lmargin=1.0in,rmargin=1.0in,bottom=1.0in,top=1.0in,twoside=False]{geometry}
\usepackage[english]{babel}
\usepackage[shortlabels,]{enumitem}
\usepackage{utf8math}
\input{insbox}
\usepackage{fontaxes}
\usepackage{trimspaces}
\usepackage{nccfoots}
\usepackage{setspace}
\usepackage{inconsolata}
\usepackage{libertine}
\usepackage{amsmath,amssymb,amsthm}
\usepackage{mathtools}
\usepackage{comment}
\usepackage[all,defaultlines=4]{nowidow}
\usepackage{microtype}
\usepackage{mathrsfs}
\usepackage{mathtools}
\usepackage{marvosym}
\usepackage{thm-restate}
\usepackage{tikz}
\usepackage{todonotes}
\usetikzlibrary{calc}
\usetikzlibrary{shapes}
\definecolor{darkgreen}{rgb}{10,117,28}
\usepackage{wrapfig}
\usepackage{diagbox}
\usepackage{stackengine}
\usepackage[ruled,vlined]{algorithm2e}
\usepackage{algorithmicx}

\definecolor{blue}{rgb}{0.1,0.2,0.5}
\definecolor{brown}{rgb}{0.6,0.6,0.2}
\usepackage[ocgcolorlinks, linkcolor={blue}, citecolor={blue}]{hyperref}
\usepackage{cleveref}

\crefformat{page}{#2page~#1#3}%
\Crefformat{page}{#2Page~#1#3}%
\crefformat{equation}{#2(#1)#3}%
\Crefformat{equation}{#2(#1)#3}%
\crefformat{figure}{#2Figure~#1#3}%
\Crefformat{figure}{#2Figure~#1#3}%
\crefformat{section}{#2Section~#1#3}
\Crefformat{section}{#2Section~#1#3}
\crefformat{chapter}{#2Chapter~#1#3}
\Crefformat{chapter}{#2Chapter~#1#3}
\crefformat{chapter*}{#2Chapter~#1#3}
\Crefformat{chapter*}{#2Chapter~#1#3}
\crefformat{part}{#2Part~#1#3}
\Crefformat{part}{#2Part~#1#3}
\crefformat{enumi}{#2(#1)#3}
\Crefformat{enumi}{#2(#1)#3}

\newtheorem{theorem}{Theorem}[section]
\crefformat{theorem}{#2Theorem~#1#3}
\Crefformat{theorem}{#2Theorem~#1#3}
\newtheorem{lemma}[theorem]{Lemma}
\crefformat{lemma}{#2Lemma~#1#3}
\Crefformat{lemma}{#2Lemma~#1#3}

\crefformat{conjecture}{#2Conjecture~#1#3}
\Crefformat{conjecture}{#2Conjecture~#1#3}

\newcommand{\newtheoremwithcrefformat}[2]{%
  \newtheorem{#1}{#2}[section]%
  \crefformat{#1}{##2\MakeUppercase#1~##1##3}%
  \Crefformat{#1}{##2\MakeUppercase#1~##1##3}%
}

\newtheoremwithcrefformat{proposition}{Proposition}
\newtheoremwithcrefformat{observation}{Observation}
\newtheoremwithcrefformat{corollary}{Corollary}
\newtheoremwithcrefformat{claim}{Claim}
\newtheoremwithcrefformat{example}{Example}
\newtheoremwithcrefformat{remark}{Remark}
\newtheoremwithcrefformat{definition}{Definition}

\makeatletter
\def\ifenv#1{
   \def\@tempa{#1}%
   \ifx\@tempa\@currenvir
      \expandafter\@firstoftwo
    \else
      \expandafter\@secondoftwo
   \fi
}
\let\wfs@comment@comment\comment
\let\comment\@undefined
\newcommand{\untoto}{\let\toto\@undefined}
\usepackage{changes}
\let\wfs@changes@comment\comment
\let\comment\@undefined

\newcommand\comment{%
    \ifthenelse{\equal{\@currenvir}{comment}}
    {\wfs@comment@comment}
    {\wfs@changes@comment}%
}
\makeatother

\usepackage{anyfontsize}

\renewcommand{\top}{\textsf{top}}

\renewcommand{\phi}{\varphi}

\newcommand{\capprox}{\ensuremath{c\mathtt{-Approx}}}
\newcommand{\PTAS}{\mathtt{PTAS}}
\newcommand{\simplify}{\ensuremath{\mathtt{Simplify}}}
\newcommand{\OPT}{\ensuremath{\text{OPT}}}
\newcommand{\T}{\ensuremath{\text{T}}}
\newcommand{\E}{\ensuremath{\text{E}}} 

\renewcommand{\leq}{\leqslant}
\renewcommand{\geq}{\geqslant}

\newcommand{\tup}[1]{\bar{#1}}

\newcommand{\R}{\mathbb{R}}

\usepackage{stmaryrd}

\newcommand{\poly}{\mathsf{poly}}

\renewcommand{\setminus}{-}

\DeclareMathOperator{\NP}{NP}
\DeclareMathOperator{\DTIME}{DTIME}
\DeclareMathOperator{\polylog}{polylog}

\usepackage{accents}

\renewcommand{\R}{\mathcal{R}}
\renewcommand{\S}{\mathcal{S}}
\renewcommand{\L}{\mathcal{L}}
\newcommand{\stabbing}{\textsc{Stabbing}\xspace}

\makeatletter
\newcommand{\@abbrev}[3]{
  \def\c@a@def##1{
      \if ##1.
        \relax
      \else
        \@ifdefinable{\@nameuse{#1##1}}{\@namedef{#1##1}{#2##1}}
        \expandafter\c@a@def
      \fi
    }
  \c@a@def #3.
}
\@abbrev{bb}{\mathbb}{ABCDEFGHIJKLMNOPQRSTUVWXYZ}
\@abbrev{bf}{\mathbf}{ABCDEFGHIJKLMNOPQRSTUVWXYZabcdefghijklmnopqrstuvwxyz}
\@abbrev{bit}{\boldsymbol}{ABCDEFGHIJKLMNOPQRSTUVWXYZabcdefghijklmnopqrstuvwxyz}
\@abbrev{cal}{\mathcal}{ABCDEFGHIJKLMNOPQRSTUVWXYZ}
\@abbrev{frak}{\mathfrak}{ABCDEFGHIJKLMNOPQRSTUVWXYZabcdefghijklmnopqrstuvwxyz}
\@abbrev{rm}{\mathrm}{ABCDEFGHIJKLMNOPQRSTUVWXYZabcdefghijklmnopqrstuvwxyz}
\@abbrev{scr}{\mathscr}{ABCDEFGHIJKLMNOPQRSTUVWXYZ}
\@abbrev{sf}{\mathsf}{ABCDEFGHIJKLMNOPQRSTUVWXYZabcdefghijklmnopqrstuvwxyz}
\@abbrev{Alg}{\Alg}{ABCDEFGHIJKLMNOPQRSTUVWXYZ}
\@abbrev{Str}{\Str}{ABCDEFGHIJKLMNOPQRSTUVWXYZ}
\@abbrev{set}{\mathbb}{ABCDEFGHIJKLMNOPQRSTUVWXYZ}
\@abbrev{tup}{\tup}{ABCDEFGHIJKLMNOPQRSTUVWXYZabcdefghijklmnopqrstuvwxyz}
\@abbrev{tld}{\tilde}{ABCDEFGHIJKLMNOPQRSTUVWXYZabcdefghijklmnopqrstuvwxyz}
\makeatother

\usepackage{wasysym}
\usepackage{xspace}

\title{A QPTAS for stabbing rectangles}

\author{
Friedrich Eisenbrand\thanks{EPFL, Switzerland, \texttt{friedrich.eisenbrand@epfl.ch}}
\and
Martina Gallato\thanks{EPFL, Switzerland, \texttt{martina.gallato@epfl.ch}}
\and
Ola Svensson\thanks{EPFL, Switzerland, \texttt{ola.svensson@epfl.ch}}
\and
Moritz Venzin\thanks{EPFL, Switzerland, \texttt{moritz.venzin@epfl.ch}}
}

\date{}

\begin{document}
\maketitle

\input{abstract}

\pagebreak
\input{intro}
\input{prel}

\input{ptas1}

\input{qptas}
\input{const}

\bibliography{citations}
\clearpage
\appendix
\newcommand{\inappendix}{yes!}
\end{document}

%% file: abstract.tex
\begin{abstract}

\noindent   
We consider the following geometric optimization problem: Given $ n $
axis-aligned rectangles in the plane, the goal is to find a set of
horizontal segments of minimum total length such that each rectangle
is \emph{stabbed}. A segment stabs a rectangle if it intersects both its
left and right edge. As such, this \emph{stabbing} problem falls into the category of
\emph{weighted geometric set cover problems} for which techniques that
improve upon the general $Θ(\log n)$-approximation guarantee have
received a lot of attention in the literature.

Chan at al.~(2018)
have shown that rectangle  stabbing is NP-hard and that it admits a constant-factor approximation
algorithm based on Varadarajan's \emph{quasi-uniform sampling method}. 

In this work we make progress on rectangle stabbing 
on two fronts. First, we
present a \emph{quasi-polynomial time approximation scheme (QPTAS)} for
rectangle stabbing. Furthermore, we provide a simple $8$-approximation
algorithm  that avoids the framework of
Varadarajan.
This settles two  open problems raised by Chan et al.~(2018).

%


%
%
\end{abstract}

%% file: intro.tex
\section{Introduction}\label{sec:intro}

We consider the following \stabbing problem: Given a set $\mathcal{R} = \{R_1, \cdots, R_n\}$ of axis-aligned rectangles in the plane, the task is to find a set of horizontal line segments of minimal total length such that all rectangles are stabbed.  A rectangle is stabbed if a line-segment in the solution intersects both its left and right edge, see Figure~\ref{fig:stab}.  This natural geometric optimization problem was introduced  by Chan et al.~\cite{Stabbing_Chan}. \stabbing can be understood as a geometric interpretation of various  combinatorial optimization problems such as \emph{frequency assignment},  \emph{message scheduling} with time-windows on a directed path, or problems in network design, see, e.g.,~\cite{Stabbing_Chan,becchetti2009latency,gmmn}.  

%
\begin{figure}[h]
	\centering
	\includegraphics[width=.40\textwidth]{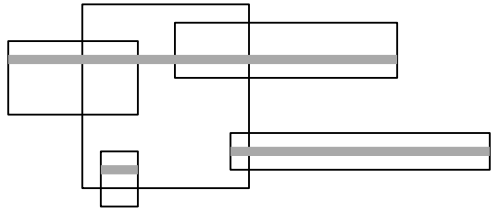}
	\caption{An instance and   solution of \stabbing}
	\label{fig:stab}
\end{figure}

\noindent 
\stabbing in turn can be interpreted as a \emph{geometric set cover problem}, in which the rectangles are the elements and the line-segments are the sets. An element (rectangle) is contained in a  set (line segment), if it is stabbed by the line segment. This immediately implies that there is a  $O(\log n )$-approximation algorithm for \stabbing~\cite{chvatal1979greedy,lovasz1975ratio}.

Improving upon the $\log n$-approximation for set cover in geometric settings has been an important area of research in computational geometry.  One successful approach~\cite{bronnimann1995almost} for \emph{unweighted} geometric  set cover is via \emph{$ε$-nets}~\cite{haussler1987}  in range-spaces of bounded \emph{VC-dimension}. Since $ε$-nets are of linear size in certain geometric settings~\cite{clarkson2007improved,matouvsek1990net,pyrga2008new}, constant factor approximation algorithms can be obtained via linear programming in these cases. 
Another very successful approach to tackle geometric set cover problems is the \emph{quasi-uniform sampling} technique of Varadarajan~\cite{Varaquasiunif}. It gives rise to a sub-logarithmic approximation algorithms for geometric set cover problems of small union complexity. The technique was   improved by Chan et al.~\cite{shallow} which then yields constant factor approximation algorithms for weighted geometric set cover problems of small \emph{shallow cell complexity}. This  is the case in the \emph{weighted disk cover} problem for example.

\medskip 
\noindent 
The state-of-the-art for \stabbing is as follows. Chan et al.~\cite{Stabbing_Chan} provide a \emph{constant-factor} approximation algorithm for \stabbing that is based on a decomposition technique and the framework of quasi-uniform sampling. More precisely, they show how to decompose \stabbing into two set cover instances of small shallow cell complexity for which the technique of Varadarajan~\cite{Varaquasiunif} and its improvement by Chan et al.~\cite{shallow} yields a constant factor approximation. The authors show furthermore that \stabbing is NP-hard via a reduction from \emph{planar vertex cover}. 

\subsubsection*{Our contribution}

In this paper, we provide the following results.  
\begin{enumerate}[i)]
\item  We show that there is a PTAS for instances of  \stabbing for which the ratio between the widths of the rectangles is bounded by a constant.
  \label{item:1}
\item This technique can then be recursively applied to yield a quasi-polynomial time approximation scheme (QPTAS) for \stabbing in general.\label{item:2} The running time of our algorithm is $n^{O(\log^3(n)/\epsilon^2)}$. This shows that \stabbing is not APX-hard unless $\NP ⊆ \DTIME(2^{\polylog(n)})$. 
 \item  We provide a simple $8$-approximation algorithm for \stabbing. First, we round the instance to a  \emph{laminar} instance, i.e., an instance in which the projections of the rectangles to the $x$-axis yields a laminar family of intervals.  This laminar instance is then  solved optimally via dynamic programming. \label{item:3}
\end{enumerate}
The contributions \ref{item:2}  and \ref{item:3} settle two open problems raised  in~\cite{Stabbing_Chan}.

\subsection*{Related work}
Approximation and hardness of approximation for geometric packing and covering problems are very active areas of research. Even for simple geometric shapes in the plane such as triangles, rectangles, disks, etc., these problems are notoriously difficult and have led to the development of a wide range of interesting techniques.

A recent breakthrough result in this area is the QPTAS for weighted independent set of rectangles by Adamaszek and  Wiese~\cite{adamaszek2013approximation}. Their approach is based on the existence of \emph{balanced cuts} with small complexity that intersect only few rectangles of the solution. This technique has proved to be very useful in various geometric settings, see, e.g.,~\cite{adamaszek2013approximation,adamaszek2019approximation,pilipczuk2020quasi}.
Interestingly, this technique was adapted by Mustafa et al.~\cite{mustafa2015quasi, mustafa2014settling} to give a QPTAS for many geometric set cover problems in the plane. This already rules out APX-hardness under reasonable complexity theoretic assumptions but obtaining a \emph{polynomial}-time approximation scheme for many of these problems remains elusive. A notable exception is the PTAS of Mustafa and Ray~\cite{mustafa2010improved} for the problem of piercing pseudo-disks. Their approach relies on local search and breaks the constant-factor barrier that is inherent in the direct application of the technique of Brönnimann and Goodrich~\cite{bronnimann1995almost}. 
On the other hand, Chan and Grant~\cite{chan2014exact} show that the set cover and the hitting set problem are already APX-hard for very simple geometric objects such as rectangles and perturbed unit squares. Indeed, for many simple objects in the plane, obtaining a polynomial time, constant-factor approximation remains an important open problem. For the problem of maximum independent set of rectangles, this was recently achieved in a breakthrough work by Mitchell~\cite{mitchell2021approximating}. 
\smallskip

To put our work into perspective, we also rely on a decomposition technique to obtain our QPTAS. However, we do not rely on the balanced cut framework. Instead we use a variation of the \emph{shifted grid} technique by Hochbaum and Maass~\cite{hochbaum1985approximation} that we combine with a simple charging scheme and careful guessing.

%% file: prel.tex

\subsection*{Notation and preliminaries}

An instance of \stabbing is specified by a set of $n$ axis-parallel rectangles $\mathcal{R} = \{R_1, \cdots, R_n\}$. Each rectangle is described by the $ y $-coordinates of its  bottom and top boundaries ($ y^b_{i} $ and $ y^t_{i} $), and by the $ x $-coordinates of its left and right boundaries ($ x^l_{i} $ and $ x^r_{i} $), i.e. $ R_i= [x^l_i, x^r_i] \times [y^b_i, y^t_i] $. 

We refer to $w(R_i):= x^r_{i}- x^l_{i}$ as the width of the rectangle $R_i$ and use $w(\R) = \max_i w(R_i)$ to denote the width of the widest rectangle in $\R$.  A solution consists of a set of horizontal line segments $\mathcal{S}$, so that each rectangle in $\mathcal{R}$ is fully \textit{stabbed} by at least one segment $ s \in \mathcal{S} $, where a segment stabs a rectangle if it intersects both its left and its right edge. The objective is to minimize the total length of the selected set of segments. For convenience, we denote by $ |\S| $ the total length of the segments in $ \S $, i.e. $ |\S| =\sum_{s \in \S} |s| $. Finally, given an instance of stabbing specified by a set of rectangles $ \R $, we denote by $ \OPT(\R) $ the total length of an optimal solution to stab all rectangles. Whenever clear, we omit $ \R $ and only write $ \OPT $.  
\newline 
It is easy to see that any solution can be transformed into a set of segments of the form $ [x^l_{i}, x^r_{j}] \times [y^t_k] $ for some $i, j, k \in [n]$,   
without increasing its total length. If a segment does not start and stop on such $ x $ coordinates, then the solution can be shortened without changing its feasibility. Furthermore, any horizontal segment can be shifted up vertically to the closest $ y^t $ without changing the set of rectangles it stabs. This gives a total number of $O(n^3)$ candidate segments. 

We will assume that all rectangles lie in a bounding box of the form $ [0,n] \times [0,2n] $. For the $ y $-coordinates, this can be achieved by passing to a combinatorially equivalent instance where $ y^b_i,y^t_j \in \{0,\dots,2n\} $ for every $ i,j \in [n] $. For the $ x $-coordinates, we can scale the instance such that the largest rectangle has width exactly $1$. Either the instance splits into independent sub-instances, or it has total width less or equal than $ n $.

Finally, when we describe our $ (1+O(\epsilon)) $-approximation schemes, we can restrict to rectangles $ R_i $ such that  $ \epsilon / n \leq w(R_i) \leq 1 $. That is, the width of each rectangle lies between $\epsilon / n$ and $1$. Indeed, since the optimal solution has cost at least $1$, we can greedily stab all rectangles of width smaller or equal than $\epsilon /n$, and remove them from the instance. This increases the total length of the solution by a factor of at most $(1+\epsilon)$.

%% file: ptas1.tex
\section{A PTAS for rectangles of similar width}\label{section:ptas}

In this section we describe a PTAS for stabbing when all rectangles are of width between $ \delta $ and $1$, which is equivalent to the bounded ratio case as described in the introduction. 
The two main technical ingredients, \cref{lemma:partitioning} and \cref{lemma:narrow}, will be used in the next section to describe a QPTAS for general instances. 

We now describe the main ideas. 
In a first step, we partition our instance $ \R $ into narrow strips using vertical lines. We will show that all rectangles that intersect one of the vertical lines can be stabbed by segments of total cost $O(\epsilon)\cdot \OPT$. We do so using a constant-factor approximation algorithm. This is \cref{lemma:partitioning}.
The resulting instances can be further partitioned into small and independent sub-instances using horizontal segments that we will include in our solution. The total length of these horizontal segments is $O(\epsilon)\cdot \OPT$. Furthermore, in each of these sub-instances, the cost of an optimal solution is $ O(w(\R)/\epsilon^2) $. This is \cref{lemma:narrow}.
When all rectangles have width between $ \delta $ and $1$, we can solve these small sub-instances by enumeration. 
\newline

We first explain how to partition our instance into narrow vertical strips of width at most $1/\epsilon$, see \cref{fig:strips}. 

\begin{lemma}\label{lemma:partitioning}[Partitioning into narrow strips]
    For any $\epsilon>0$, given an instance $\R$ of \stabbing,  we can  in polynomial time find a set $ \S_E $ of line segments such that 
    \begin{itemize}
        \item the total length, i.e., the cost of the line segments in $\S_E$ is $16\epsilon \cdot \OPT$; and
        \item after removing the rectangles stabbed by segments in $ \S_E $, the instance splits into disjoint strips, each of width at most $w(\R)/\epsilon$. 
    \end{itemize}
\end{lemma}
\begin{figure}[h]
	\centering
	\includegraphics[width=1.0\textwidth]{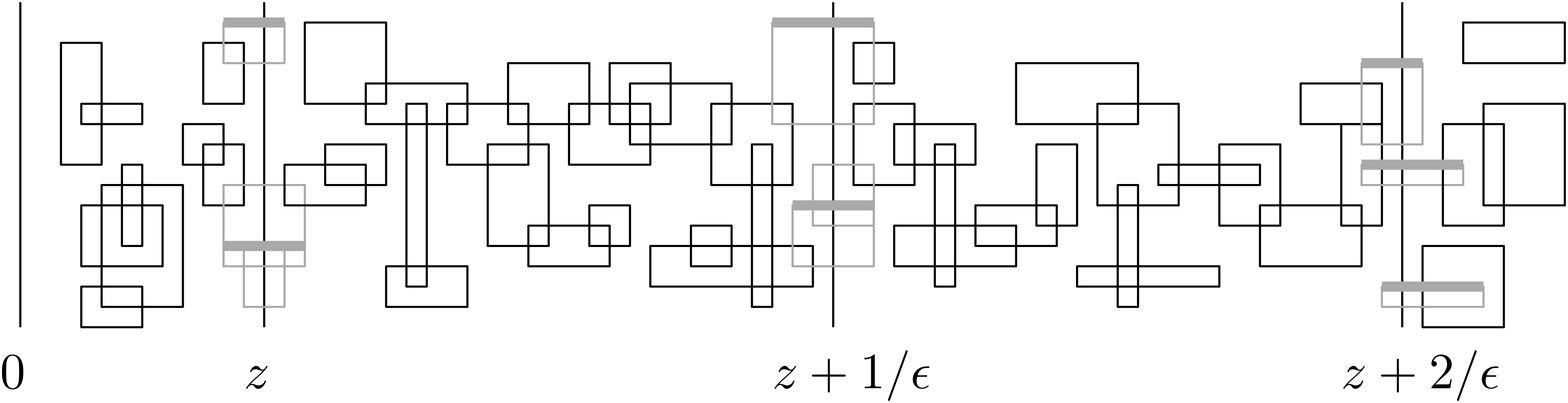}
	\caption{Vertical partitioning into strips. Grey rectangles are covered by segments in $ \S_E $.}
	\label{fig:strips}
\end{figure}

\begin{proof}
	
	We first note that, by scaling the instance, we may assume without loss of generality that $w(\R) = 1$. 
    For some number $z \in [0, 1/\epsilon) $, a multiple of $ \epsilon/n $, to be specified later, we intersect the bounding box of $\R$ with the set of vertical lines $\mathcal{L}_z$ with $x$-coordinates $z + i/\epsilon$ for $i\in \mathbb{Z}$. 
	We partition our instance as follows. We consider the rectangles that are intersected by some line from $ \mathcal{L}_z $, and call this set of rectangles $E_z$. We will show that we can stab all rectangles from $ E_z $ using segments of total length at most $ O(\epsilon)\cdot \OPT $. 
	The remaining instance consists of rectangles that are split into disjoint strips, each of width at most $(1/\epsilon)$. It remains to show that there exists a choice of $z$, for which there is a set of segments $ \mathcal{S}_{E_z} $ stabbing all rectangles in $ E_z $ of total length at most $ O(\epsilon )\cdot\OPT$. This is sufficient: using the $8 $-approximation algorithm for stabbing from \cref{sec:const} and trying out all $ n/\epsilon^2 $ possibilities for $ z $, we can find a set of segments $ \S_E $ of cost at most $ 16 \epsilon \cdot \OPT $. 
	
	
	To this end, denote by $\mathcal{S}_{\OPT}$ the set of segments in the optimal solution for the entire instance.
	Fix $s \in \mathcal{S}_{\OPT}$ and an integer $z$. If $s$ intersects one of the vertical lines from $ \mathcal{L}_z $, then, for each intersection point $(p_1, p_2)$, add the segment $[p_1 - 1, p_1 + 1]\times p_2$ to $\mathcal{S}_{E_z}$. Since $\mathcal{S}_{\OPT}$ covers all rectangles, $\mathcal{S}_{E_z}$ will stab all rectangles in $E_z$.  Moreover, it holds that: 	
	\begin{displaymath}
	\sum_{s \in \mathcal{S}_{E_z}}|s| \leq \sum_{s \in \mathcal{S}_{\OPT}} 2 \cdot |s \cap \mathcal{L}_z|.
	\end{displaymath}
	Whenever a segment $s \in \OPT$ is larger than $1/\epsilon$, the contribution of the corresponding term is at most $2\epsilon |s|$, for any $z \in [0, 1/\epsilon)$. Therefore, we have:
	\begin{displaymath}
	\sum_{s \in \mathcal{S}_{E_z}}|s| \leq \sum_{s \in \mathcal{S}_{\OPT}\,:\, |s| < 1/\epsilon} 2\cdot |s \cap \mathcal{L}_z| + \sum_{s \in \mathcal{S}_{\OPT}\,:\, |s| \geq 1/\epsilon} 2\epsilon\cdot|s|.
	\end{displaymath}

	If we sample a multiple $ z $ of $ \epsilon/n $ from the interval $ [0,1/\epsilon) $ uniformly at random, the probability that a segment $ s \in \mathcal{S}_{\OPT}$ with $|s| < 1/\epsilon  $ is hit by $\L_z$ is at most $2 \epsilon |s| $. Indeed, if $ z $ is sampled uniformly from $ [0,1/\epsilon) $, then the probability is equal to $ \epsilon |s| $, and our discretization increases this by a factor of $ 2 $ at most. By linearity of expectation, it follows that
	
	
	\begin{displaymath}
	\mathbb{E}_{z}\big[\sum_{s \in \mathcal{S}_{E_z}}|s|\big] \leq 2\epsilon\cdot \OPT .
	\end{displaymath}

\end{proof}

We now explain how to further simplify instances that are contained in a narrow strip. 


\begin{lemma}\label{lemma:narrow}[The case of narrow strips]
	Consider a \stabbing instance $\R$ in which all  rectangles are contained in a strip of width $w(\R)/\epsilon$.  We can in polynomial time find a set of line segments $ \S_h $ of cost  $\epsilon \cdot OPT $ such that, after removing those rectangles stabbed by $\S_h$, the remaining rectangles are partitioned into independent sub-instances $ \R_1, \dots, \R_k $ such that $ \OPT(\R_i) \leq O(w(\R)/\epsilon^2) $ for each $i$. 
\end{lemma}

\begin{wrapfigure}[27]{L}{0.2\textwidth}
	\vspace{-10.8pt}
	\caption{Strip decomposition using horizontal lines. \phantom{éalsdfjk}}\label{wrap-fig:1}
	\includegraphics[width=0.2\textwidth]{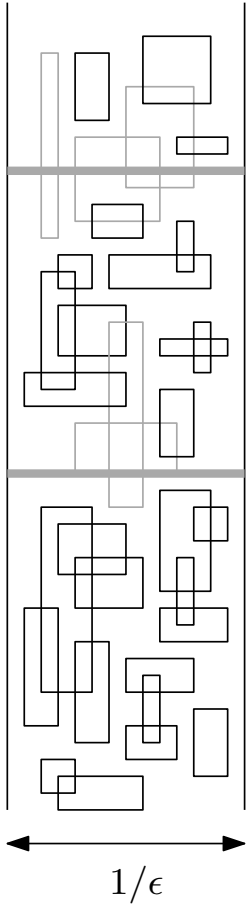}
\end{wrapfigure}

\bigskip
\noindent \textit{Proof.} By translating and preprocessing (see the preliminaries) we may assume that all rectangles are contained inside $[0, w(\R)/\epsilon]\times [0, 2n]$ and that their $y$-coordinates take integer values. We start by sweeping from bottom to top a horizontal line. We pick the smallest value of $z$, say $z_1$, so that a $c$-approximation algorithm for stabbing, for example the $ 8 $-approximation algorithm from \cref{sec:const}, requires segments of total length strictly greater than $ c w(\R)/\epsilon^2 $ to stab all rectangles entirely contained inside the region $[0, w(\R)/\epsilon]\times [0,z_1]$. Observe that this value can only increase by $w(\R)/\epsilon $ when going from $ z $ to $ z+1 $.

It follows that the cost of an optimal solution to cover rectangles in $[0, w(\R)/\epsilon]\times [0,z_1]$ is between $w(\R)/\epsilon^2$ and $ c \cdot w(\R)/\epsilon^2 + w(\R)/\epsilon$. We add the segment $[0, w(\R)/\epsilon] \times z_1$ to the set $ \S_h $ and remove stabbed rectangles from our instance.  We repeat this procedure until $[0, w(\R)/\epsilon]\times [0, 2n]$ is subdivided into independent instances $ \R_1, \dots, \R_k $ with $w(\R)/\epsilon^2 \leq \OPT(\R_i)\leq c \cdot w(\R)/\epsilon^2 + w(\R)/\epsilon $, where $\R_i$ consists of those rectangles contained in $[0, 1/\epsilon]\times [z_{i-1}, z_i]$ and $ z_0=0, z_{k}=2n $. This is illustrated if \cref{wrap-fig:1}.
	
Since the $ \R_i $ are independent sub-instances of $\R$, we have
	\begin{displaymath}
	\sum_{i=1}^{k} \OPT(\R_i) \leq \OPT(\R).
	\end{displaymath}
	On the other hand, since $\OPT(\R_i) \geq  w(\R)/\epsilon^2$, we have that
	\begin{displaymath}
	\sum_{s \in \S_h}|s| \leq \sum_{i=1}^{k-1}\epsilon\cdot \OPT(\R_i) \leq \epsilon \cdot \OPT(\R).
	\end{displaymath}
	This shows that adding the segments in $ \S_h $ to our solution incurs a cost of $\epsilon \cdot \OPT$. \qed


\bigskip
To devise our algorithms, it will be convenient to state the following corollary, which follows by first applying \cref{lemma:partitioning} and then applying \cref{lemma:narrow} on each of the strips.
\begin{corollary}
    For every $\epsilon >0$,  given a \stabbing instance $\R$, we can in polynomial time find a set $\S$ of line segments of total cost $17\epsilon \cdot \OPT$ such that, after removing those rectangles stabbed by $\S$, the remaining rectangles are partitioned into independent sub-instances $\R_1, \ldots, \R_m$ such that $\OPT(\R_i) \leq O( w(\R)/\epsilon^2)$ for each $i$.
    \label{corollary:splitting_up}
\end{corollary}

It is now immediate to obtain a polynomial time approximation scheme for instances where all rectangles are of width between $\delta$ and $1$. Using the set $\S$ of line segments of total length $O(\epsilon)\cdot \OPT$ given by \cref{corollary:splitting_up}, we partition our original instance $ \R $ into at most $ n $ non-empty sub-instances $ \R_1, \ldots, \R_m $ where $ \OPT(\R_i) \leq O(w(\R)/\epsilon^2) = O(1/\epsilon^2)$ for each $ i $. 
Since all rectangles are of width greater than $ \delta $, the optimal solution to each sub-instance can only use at most $O(1/(\delta \cdot \epsilon^2))$ segments. As there are at most $ O(n^3) $  candidate segments (see preliminaries), this gives at most $ n^{O(1/(\delta \cdot \epsilon^2))} $ possible feasible solutions for each $\R_i$. We enumerate all possibilities and pick a solution of minimal total length. Since $ |\S| = O(\epsilon)\cdot \OPT(\R)  $ and  $ \sum_{i=1}^{m} \OPT(\R_i) \leq \OPT(\R) $, the union of $ \S$ and the respective segments from the optimal solution for each $ \R_i $ yield an $ (1+O(\epsilon)) $-approximation. 

\begin{lemma}
There is a polynomial time approximation scheme (PTAS) for instances of \stabbing where all rectangles have width between $\delta$ and $1$, for $\delta$ a constant.
\end{lemma}

%% file: qptas.tex
\section{A QPTAS for \stabbing}\label{sec:qptas}

In this section, we give a quasi-polynomial time approximation algorithm for stabbing.

\begin{theorem}
    For every $\epsilon >0$, there is a  $(1+\epsilon)$-approximation algorithm for \stabbing that runs in time $n^{O(\log^3(n)/\epsilon^2)}$.
    \label{thm:qptas}
\end{theorem}
The main idea is to recursively apply \cref{corollary:splitting_up} and \emph{guess} some long segments of the optimal solution. This divides the current instance into sub-instances with no long rectangles, that we can solve recursively independently from each other. Specifically, given an instance of \stabbing $\R$ and a parameter $\mu>0$, the algorithm proceeds as follows:

\begin{itemize}
    \item  Apply \cref{corollary:splitting_up} (with the $\epsilon$ of the corollary set to $\mu$) to find a set $\S$ of line segments of cost at most $17\mu \cdot \OPT(\R)$ such that the remaining unstabbed rectangles are partitioned into independent sub-instances $\R_1, \ldots, \R_m$ with $\OPT(\R_i) = O(w(\R)/ \mu^2)$. 
    \item Each $\R_i$ is then independently solved as follows. Guess all line segments $\S_{guess}$ in an optimal solution of $\R_i$ that have length at least $ w(\R)/2$ and let  $\R_i^{res}$ be the sub-instance of $\R_i$ without the rectangles stabbed by $\S_{guess}$. If $\R_i^{res}$ is non-empty, solve it recursively.
\end{itemize}

Some comments regarding the second step are in order. First, the algorithm does not have access to an optimal solution to the instance $\R_i$ and so it cannot guess the line segments $\S_{guess}$. Instead, this is solved by enumeration. Indeed, since $\OPT(\R_i) \leq O(w(\R)/\mu^2)$, an optimal solution to $\R_i$ may only place $O(1/\mu^2)$ line segments of length at least $ w(\R)/2$. As we did in the polynomial time approximation scheme, we can enumerate over all possibilities by considering $n^{O(1/\mu^2)}$ many options and then take the best found solution. Second, for a correct guess $\S_{guess}$, we must have that all rectangles in $\R_i$ of width at least $ w(\R)/2$ have been stabbed by $\S_{guess}$. We can therefore restrict ourselves to only recurse on instances with  $w(\R_i^{res}) \leq w(\R)/2$.  In other words, in each recursive call we have at least halved the width of the widest rectangle. As we may assume that the ratio between any two rectangles is $n/\epsilon$ (see preliminaries), we have that the recursion depth $H$ is upper bounded by $O(\log(n/\epsilon)) = O(\log n)$. The following two lemmas therefore imply  \cref{thm:qptas} by selecting $\mu = O(\epsilon/\log(n))$.

\begin{lemma}
    The running time is upper bounded by $n^{ O(H/\mu^2)}$.
\end{lemma}
\begin{proof}
This follows since the recursive execution can be illustrated by a tree of height $H$ where each node corresponds to one call to the algorithm. In each such call, the algorithm first uses \cref{corollary:splitting_up} to obtain instances $\R_1, \ldots, \R_m$. Then, for each $\R_i$ independently, it performs $n^{O(1/\mu^2)}$ guesses of $\S_{guess}$ followed by a recursive call. Hence, the degree of the tree is  at most $m \cdot n^{O(1/\mu^2)} = n^{O(1/\mu^2)}$   using $m\leq n$. This implies that the number of nodes of the tree is $n^{ O(H/\mu^2)}$, which in turn also upper bounds the running time of the algorithm.
\end{proof}

\begin{lemma}
    The algorithm outputs a $(1+17(H+1) \mu)$-approximate solution.
\end{lemma}
\begin{proof}
    We prove the statement by induction on $H$.  In the base case $H=0$, no recursive calls are made. In that case the algorithm first finds the set $\S$ of line segments of cost at most $17\mu \cdot \, \OPT(\R)$. It then optimally solves the sub-instances $\R_1, \ldots, \R_m$ to get a solution of total cost $17\mu \cdot\OPT(\R) + \sum_{i=1}^m\OPT(\R_i) \leq (1+17\mu) \cdot\OPT(\R) $, where we used that $\sum_{i=1}^m \OPT(\R_i) \leq \OPT(\R)$.
    
    The inductive case follows similarly: the algorithm first finds the set $\S$ of line segments of cost at most $17\mu \cdot\OPT(\R)$. It then recursively solves the sub-instances $\R_1, \ldots, \R_m$, where, by the induction hypothesis, the solution to $\R_i$ has cost at most $(1+17H \mu) \cdot\OPT(\R_i)$ (for the correct guess $\S_{guess}$). It follows that total  cost of the solution is at most $17\mu \cdot \OPT(\R) + \sum_{i=1}^m (1+17 H\mu) \cdot\OPT(\R_i) \leq (1+ 17(H+1)\mu) \cdot\OPT(\R) $.
\end{proof}

%% file: const.tex
\section{A simple $8$-approximation for \stabbing}\label{sec:const}

In this section we describe a very simple constant-factor approximation algorithm for \stabbing. Our algorithm is based on a shifting technique and dynamic programming. In particular, it avoids the quasi-uniform sampling approach as outlined in \cite{Stabbing_Chan}. Other than being very involved, quasi-uniform sampling yields approximation guarantees which are generally quite large. In contrast to this, our algorithm exploits the underlying geometry of the problem in a more direct way and achieves an approximation factor of $8$. 

We now explain our approach. In a first step, we will show that we can solve \emph{laminar} instances of stabbing \emph{exactly}. This can be done in $O(n^5)$ time using dynamic programming. In a second step, by shifting and scaling the rectangles, we transform any instance of stabbing into a laminar one. The solution to the laminar instance can then be scaled to yield a solution to the original instance. This only increases the total length of the segments by a factor of $8$.

We now give the definition of a laminar instance. 
\begin{definition}
	An instance of stabbing is laminar if its projection onto the $ x $-axis is a laminar family of intervals, that is, any two of these intervals are either disjoint or one is contained in the other. 
\end{definition}

Our approach exploits the laminarity in three simple ways:
\begin{enumerate}
	\item Any subset of the rectangles is still laminar. \label{laminar:one}
	\item The rectangle of \emph{largest} width need only be covered by a segment of length equal to its width.\label{laminar:two}
	\item Stabbing the rectangle of largest width results in $4$ \emph{independent} instances.\label{laminar:three}
\end{enumerate}

Property \cref{laminar:one} is immediate from the definition of laminarity. For \cref{laminar:two}, since the projection onto the $x$-axis is laminar, it is easy to see that the segment stabbing the rectangle $W$ of largest width can be assumed to start and end at $x_W^l$ and $x_W^r$ respectively. Finally, we illustrate \cref{laminar:three}, which is at the heart of our dynamic program, in the following figure.

\begin{figure}[h]
	\centering
	\includegraphics[width=1\textwidth]{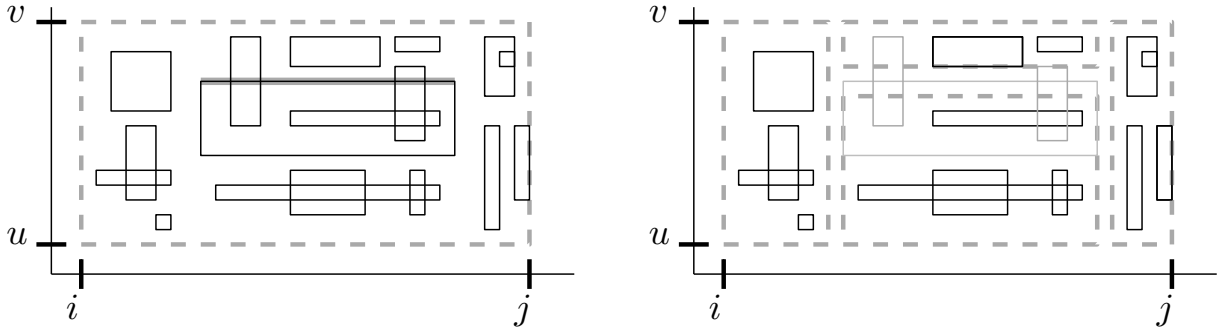}
	\caption{On the left panel we cover the largest rectangle by a segment. The resulting instance splits into four independent sub-instances, as shown on the right.  }
	\label{fig:laminar-before-after}
\end{figure} 

Specifically, suppose we are given a (sub-)instance of \stabbing consisting of all rectangles that are contained inside the bounding box $B:=[i,j]\times[u,v]$. Furthermore, suppose that we decide to stab the largest rectangle $W$ contained in $B$ by the segment $s := [x_W^l, x_W^r]\times y$. All rectangles in $B$ that are not covered by $s$ must be contained in either $B_1:= [i, x_W^l]\times [u,v], B_2 := [x_W^r, j]\times [u,v], B_3 := [x_W^l, x_W^r]\times[u,y-1]$ or $ B_4 := [x_W^l, x_W^r]\times[y+1, v] $. On the other hand, all rectangles that are stabbed by $s$ are not contained in either $B_1, B_2, B_3$ or $B_4$. This is illustrated on the right-hand side of \cref{fig:laminar-before-after} and proves \cref{laminar:three}. \\
Using the same notation as above, it is now very simple to state our dynamic program. For $i,j \in \{0, 1, \ldots, n\}$ and $u, v\in \{0, 1, \ldots, 2n\}$, our state space consists of (sub-)instances $I(i,j,u,v)$ where we only consider rectangles contained inside $[i,j]\times[u,v]$. To compute the optimal value (and the respective segments) for stabbing rectangles in $I(i,j,u,v)$ we proceed as follows. We find the rectangle $W$ of largest width inside $I(i,j,u,v)$. For all segments that stab $W$, i.e. segments contained in the set 
\begin{displaymath}
	\mathcal{Y}_W:= \{ [x_W^l, x_W^r]\times y \,|\,  y_{W}^b \leq y \leq y_{W}^t \text{ and } y \in \{0, \ldots, 2n\} \},
\end{displaymath}
we look up the cost of the optimal solutions of the four respective sub-instances that are created with such a segment, and pick the one such that the overall cost is minimized. More specifically,

\begin{equation}\label{eq:dp}\tag{DP}
\begin{split}
\OPT(B)= |x_W^r-x_W^l|+\OPT(B_1)+\OPT(B_2)+ \min_{y \in \mathcal{Y}_W}\big(\OPT(B_3)+\OPT(B_4)\big).
\end{split}
\end{equation}

\noindent
Observe that each sub-instance called in the above recursion has strictly fewer rectangles than the original instance $I(i,j,u,v)$. Since we can solve (sub-)instances containing at most $1$ rectangle exactly, correctness follows from \cref{eq:dp}. Finally, as the number of states is of size $O(n^4)$, we can conclude.

\begin{lemma}
In $O(n^5)$ time we can compute an optimal solution to a laminar instance of \stabbing.
\end{lemma}

We show how to use the algorithm for laminar instances to obtain an 8-approximation for the general case. To this end, we transform a general instance of stabbing into a laminar one as follows. 
We first scale, with respect to the $ x $-axis, each rectangle $ R_i $ to the right, such that its width is a power of $ 2 $. That is, if $ 2^{t-1} < w_{i} \leq 2^t $ for some integer $ t $, the new width $ w_i' $ becomes $ 2^t $. Then, we move each rectangle to the left, to the nearest multiple of its width. That is, we shift each rectangle $ R_i $ such that $ x_{i}^l $ becomes $ k w_{i}' $, where $ k=\max \{z \in \mathbb{Z}: z \cdot \ w_{i}' \leq x_i^l\} $. After this scaling and moving procedure we get a laminar family of rectangles $ \R' $. 
\begin{figure}[h]
	\centering
	\includegraphics[width=1\textwidth]{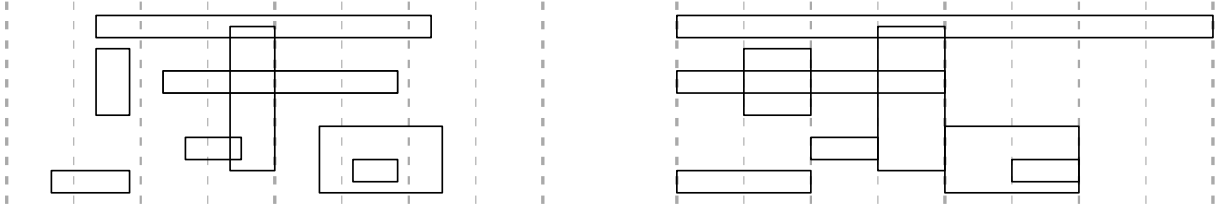}
	\caption{On the left, a general instance. On the right, the laminar instance that we obtain after the stretching and shifting procedure. The dashed gray lines denote the respective powers of $2$.}
	\label{fig:decomp}
\end{figure}

\noindent
Observe that any solution for this instance $ \R' $, when scaled by a factor of $ 2 $ to the right, yields a feasible solution for the original instance $ \R $. On the other hand, each segment in a feasible solution for $ \R $ needs to be scaled by at most a factor of $ 4 $ to get a feasible solution for $ \R' $. Indeed, for any segment $ s $ in a feasible solution for $ \R $ of the form $ s=[x_i^l,x_j^r] \times y $, with $ 2^{t-1} < x_j^r-x_i^l \leq 2^t  $, the segment $ s' $ that is obtained from $ s $ by rounding down (up) $ x_i^l $ ($ x_i^r $) to the next multiple of $ 2^t $  will stab all rectangles in $ \R' $ corresponding to rectangles from $ \R $ that were stabbed by $ s $. The length of each segment only increases by a factor of $ 4 $. 
In summary, we lose a factor of $4$ when passing to the laminar instance, and a factor of $2$ when going from a solution for the laminar instance to a feasible solution for the original one. This yields the following. 

\begin{theorem}
	In $O(n^5)$ time, we can find an $8$-approximation to \stabbing. 
\end{theorem}